\documentclass[letterpaper,accepted=2023-09-04]{quantumarticle}
\pdfoutput=1

\usepackage[T1]{fontenc}
\usepackage[latin9]{inputenc}
\usepackage{babel}
\usepackage{float}
\usepackage{units}
\usepackage{amsmath}
\usepackage{amsthm}
\usepackage{amssymb}
\usepackage{graphicx}
\usepackage[unicode=true]{hyperref}
\usepackage{xcolor}

\PassOptionsToPackage{compress}{natbib}
\usepackage[numbers]{natbib}

\makeatletter

\let\newfloat\newfloat@ltx
\makeatother

\floatstyle{ruled}
\newfloat{algorithm}{tbp}{loa}
\providecommand{\algorithmname}{Algorithm}
\floatname{algorithm}{\protect\algorithmname}

\numberwithin{equation}{section}
\theoremstyle{plain}
\newtheorem{fact}{\protect\factname}[section]
\theoremstyle{plain}
\newtheorem{thm}{\protect\theoremname}[section]
\theoremstyle{remark}
\newtheorem{claim}{\protect\claimname}[section]

\date{}

\makeatother

\providecommand{\claimname}{Claim}
\providecommand{\factname}{Fact}
\providecommand{\theoremname}{Theorem}

\newcommand{\pick}{\stackrel{\$}{\leftarrow}}

\DeclareMathOperator*{\E}{\mathbb{E}}

\begin{document}
\title{Forging quantum data: classically defeating an $\mathsf{IQP}$-based
quantum test}
\author{Gregory D. Kahanamoku-Meyer}
\affiliation{Department of Physics, University of California at Berkeley, Berkeley, CA 94720}
\email{gkm@berkeley.edu}
\maketitle

\begin{abstract}
Recently, quantum computing experiments have for the first time exceeded the capability of classical computers to perform certain computations---a milestone termed ``quantum computational advantage.''
However, verifying the output of the quantum device in these experiments required extremely large classical computations.
An exciting next step for demonstrating quantum capability would be to implement tests of quantum computational advantage with efficient classical verification, such that larger system sizes can be tested and verified.
One of the first proposals for an efficiently-verifiable test of quantumness consists of hiding a secret classical bitstring inside a circuit of the class IQP, in such a way that samples from the circuit's output distribution are correlated with the secret.
The classical hardness of this protocol has been supported by evidence that directly simulating IQP circuits is hard, but the security of the protocol against other (non-simulating) classical attacks has remained an open question.
In this work we demonstrate that the protocol is not secure against classical forgery.
We describe a classical algorithm that can not only convince the verifier that the (classical) prover is quantum, but can in fact can extract the secret key underlying a given protocol instance.
Furthermore, we show that the key extraction algorithm is efficient in practice for problem sizes of hundreds of qubits.
Finally, we provide an implementation of the algorithm, and give the secret vector underlying the ``\$25 challenge'' posted online by the authors of the original paper.
\end{abstract}

\section{Introduction}

Recent experiments have demonstrated groundbreaking quantum computational power in the laboratory, showing \emph{quantum computational advantage}~\cite{arute_quantum_2019, zhong_quantum_2020, wu2021strong, zhu2021quantum}.
In the past decade, much theoretical work has gone into designing experimental protocols expressly for this purpose, and providing evidence for the classical hardness of reproducing the experimental results~\cite{aaronson_computational_2011, farhi2016quantum, lund2017quantum, harrow2017quantum, boixo2018characterizing, terhal2018quantum, neill_blueprint_2018, bravyi_quantum_2018, bravyi_quantum_2019, bouland_complexity_2019, aaronson2019classical, brakerski_cryptographic_2019, brakerski2020simpler, kahanamoku2021classically}.
A difficulty with many of them, however, is that the quantum machine's output is hard to \emph{verify}.
In many cases, the best known algorithm for directly checking the solution is equivalent to classically performing the computational task itself.
This presents challenges for validation of the test's results, because an ideal demonstration of quantum advantage occurs in the regime where a classical solution is not just difficult, but impossible with current technology.
In that regime, experiments have had to resort to indirect methods to demonstrate that their devices are producing correct results~\cite{arute_quantum_2019, zhong_quantum_2020, wu2021strong, zhu2021quantum}.

\begin{figure}
\begin{raggedright}
\includegraphics[width=0.9\columnwidth]{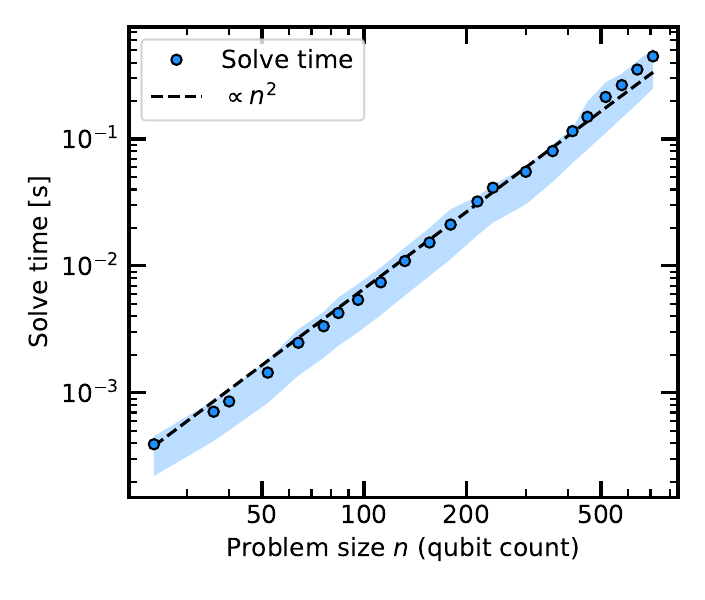}
\par\end{raggedright}
\caption{\label{fig:solvetime}Mean time to extract the secret vector $\boldsymbol{s}$
from $X$-programs constructed as described in \cite{shepherd_temporally_2009}.
Shaded region is the first to third quartile of the distribution of
runtimes. We observe that the time is polynomial and fast in practice
even up to problem sizes of hundreds of qubits. See Section \ref{subsec:Implementation}
for a discussion of the $\mathcal{O}\left(n^{2}\right)$ scaling.
The data points were computed by applying the algorithm to 1000 unique
$X$-programs at each problem size. The secret vector was successfully
extracted for every $X$-program tested. Experiments were completed
using one thread on an Intel 8268 ``Cascade Lake''
processor.}
\end{figure}

In 2009, an efficiently-verifiable test of quantum computational advantage was proposed based on ``instantaneous quantum polynomial-time'' (IQP) circuits---quantum circuits in which all operations commute~\cite{shepherd_temporally_2009}.
The protocol places only moderate requirements on the quantum device, making it potentially a good candidate for near-term hardware.
Furthermore, later papers showed based on reasonable assumptions that classically sampling from the resulting distribution should be hard \cite{bremner_classical_2011, bremner_average-case_2016}.
This suggests that a ``black-box'' approach to cheating classically (by simply simulating the quantum device) is indeed computationally hard, and only a couple hundred qubits would be required to make a classical solution intractable.

Importantly, however, the classical verifier of the efficiently-verifiable protocol does not explicitly check whether the prover's samples come from the correct distribution (in fact, doing such a check efficiently is probably not possible \cite{bremner_classical_2011}).
Instead, the sampling task is designed such that bitstrings from its distribution will be orthogonal to some secret binary vector $\boldsymbol{s}$ with high probability, and it is this property that is checked by the verifier.
A question that has remained open is whether a classical machine can efficiently generate samples satisfying the orthogonality check, without necessarily approximating the actual circuit's distribution.
In this work we show that the answer to this question is yes.
We give an explicit algorithm that can extract the secret bistring $\boldsymbol{s}$ underlying any instance of the protocol, thus making it trivial to generate orthogonal samples that pass the verifier's test.
The main results described here are a statement of the algorithm, a proof that a single iteration of it will extract the secret vector $\boldsymbol{s}$ with probability $\nicefrac{1}{2}$ (which can be made arbitrarily close to 1 by repetition), and empirical results demonstrating that the algorithm is efficient in practice (summarized in Figure \ref{fig:solvetime}). 

The following is a summary of the paper's structure. In Section \ref{sec:Background}, we review the protocol's construction and some relevant analysis from the original paper.
In Section \ref{sec:Algorithm} we describe the algorithm to extract the secret key, and therefore break the protocol's security against classical provers.
There we also discuss briefly our implementation of the algorithm.
In Section \ref{sec:Discussion} we discuss the results, and provide the secret key underlying the ``\$25 challenge'' that accompanied the publication of the protocol.

\section{Background}
\label{sec:Background}

\paragraph{Overview of protocol}

Here we summarize the IQP-based protocol for quantum advantage, in the standard cryptographic terms of an ostensibly quantum \emph{prover} attempting to prove its quantum capability to a classical \emph{verifier}.
We refer the reader to the work that proposed the protocol for any details not covered here~\cite{shepherd_temporally_2009}.
The core of the protocol is a sampling problem.
The verifier generates a Hamiltonian $H_{P}$ consisting of a sum of products of Pauli $X$ operators, and asks the quantum prover to generate samples by measuring the state $e^{iH_{P}\theta}\left|0^{\otimes n}\right\rangle$ for some value of the ``action'' $\theta$.
The Hamiltonian $H_P$ is designed such that the measured bitstrings $\left\{ \boldsymbol{x}_{i}\right\} $ are biased with respect to a secret binary vector $\boldsymbol{s}$, so that $\boldsymbol{x}_{i}\cdot\boldsymbol{s}=0$ with high probability (where $(\cdot)$ represents the binary inner product, modulo 2).
The classical verifier, with knowledge of $\boldsymbol{s}$, can quickly check that the samples have such a bias. 
Since $\boldsymbol{s}$ should be only known to the verifier, it was conjectured
that the only efficient way to generate such samples is by actually
computing and measuring the quantum state~\cite{shepherd_temporally_2009}.
However, in Section \ref{sec:Algorithm} we show that it is possible to extract $\boldsymbol{s}$ classically from just the description of the Hamiltonian.

\paragraph{X-programs}

A Hamiltonian of the type used in this protocol can be described by a rectangular matrix of binary values, for which each row corresponds to a term of the Hamiltonian.
Given such a binary matrix $P$ (called an ``$X$-program''), the Hamiltonian is
\begin{equation}
H_{P}=\sum_{i}\prod_{j}X^{P_{ij}}\label{eq:hamiltonian}
\end{equation}
In words, a 1 in $P$ at row $i$ and column $j$ corresponds to the inclusion of a Pauli $X$ operator on the $j^{\mathrm{th}}$ site in the $i^{\mathrm{th}}$ term of the Hamiltonian. The $X$-program also has one additional parameter $\theta$, which is the ``action''---the integrated energy over time for which the Hamiltonian will be applied.
For the protocol relevant to this work, the action is set to $\theta = \pi/8$ (see below).

\paragraph{Embedding a bias and verifying the output}

In order to bias the output distribution along $\boldsymbol{s}$, a submatrix with special properties is embedded within the matrix $P$.
Notationally, for a vector $\boldsymbol{s}$ and matrix $P$, let the submatrix $P_{\boldsymbol{s}}$ be that which is generated by deleting all rows of $P$ that are orthogonal to $\boldsymbol{s}$.
Letting $\boldsymbol{X}$ represent the distribution of measurement results for a given $X$-program, it can be shown that the probability that a measurement outcome is orthogonal to the vector $\boldsymbol{s}$, $\Pr[\boldsymbol{X}\cdot \boldsymbol{s} = 0]$, depends only on the submatrix $P_{\boldsymbol{s}}$.
The rows of $P$ that are orthogonal to $\boldsymbol{s}$ are irrelevant.
The protocol uses that fact to attempt to hide $P_{\boldsymbol{s}}$ (and thus $\boldsymbol{s}$): starting with a matrix $P_{\boldsymbol{s}}$ that produces a bias, we may attempt to hide it in a larger matrix $P$ by appending rows that are random aside from having $\boldsymbol{p} \cdot \boldsymbol{s}=0$, and then scrambling the new, larger matrix in a way that preserves the bias.

But what matrix $P_{\boldsymbol{s}}$ should one start with?
In the protocol, the verifier sets $P_{\boldsymbol{s}}$ to the generator matrix for a binary code of block length $q \equiv 7 \pmod 8$ whose codewords $\boldsymbol{c}$ have $\mathrm{wt}(\boldsymbol{c}) \in \{-1, 0\} \pmod 4$ (and both those weights are represented, that is, the codewords do not all have weight $0 \pmod 4$).
In \cite{shepherd_temporally_2009}, the authors suggest specifically using a binary quadratic residue (QR) code because it has the desired codeword weights.
The action $\theta$ is set to $\pi/8$.
As described in Facts~\ref{fact:quantum-strategy} and~\ref{fact:classical-strategy} below, this configuration leads to a gap between the quantum and classical probabilities of generating samples orthogonal to $\boldsymbol{s}$ (for the best known classical strategy before this work).
The verifier's check is then simply to request a large number of samples, and determine if the fraction orthogonal to $\boldsymbol{s}$ is too large to have likely been generated by any classical strategy.

In the two Facts below, we recall the probabilities corresponding to the quantum strategy and previously best-known classical strategy~\cite{shepherd_temporally_2009}.
The reasoning behind the classical strategy (Fact~\ref{fact:classical-strategy}) forms the setup for the new algorithm described in this paper; it is worth understanding its proof before moving on to the algorithm in Section \ref{sec:Algorithm}.

\begin{fact}
\label{fact:quantum-strategy}
\textbf{Quantum strategy} 

Let $P$ be an $X$-program constructed by embedding a submatrix with the properties described above.
Let $\boldsymbol{X}$ be a random variable representing the distribution of bitstrings from an $n$-qubit quantum state $e^{iH_{P}\pi/8}\left|0\right\rangle $ measured in the $Z$ basis, where $H_{P}$ is defined as in Equation~\ref{eq:hamiltonian}. Then,
\begin{equation}
\Pr\left[\boldsymbol{X}\cdot\boldsymbol{s}=0\right]=\cos^{2}\left(\frac{\pi}{8}\right)\approx0.85\cdots
\label{eq:quantum-strat}
\end{equation}
\end{fact}

\begin{proof}
The entire proof is contained in \cite{shepherd_temporally_2009}.
To summarize, it is shown that for any string $\boldsymbol{z}$ and corresponding submatrix $P_{\boldsymbol{z}}$, the probability is
\begin{equation}
\Pr\left[\boldsymbol{X}\cdot\boldsymbol{z}=0\right] = \E_{\boldsymbol{c}} \left[ \cos^{2}\left( \theta \cdot (q - 2\mathrm{wt}(\boldsymbol{c}) \right) \right]
\end{equation}
where $\theta$ is the action, $q$ is the number of rows in $P_{\boldsymbol{z}}$ and the expectation is taken over the codewords $\boldsymbol{c}$ of the code generated by the submatrix $P_{\boldsymbol{z}}$.
When the values of $\theta = \pi/8$, $q \equiv 7 \pmod 8$ and $\mathrm{wt}(\boldsymbol{c}) \in \{-1, 0\} \pmod 4$ corresponding to the specific submatrix $P_{\boldsymbol{s}}$ are substituted into this expression, the result is Equation~\ref{eq:quantum-strat}.

\end{proof}

\begin{fact}
\label{fact:classical-strategy}\textbf{Classical strategy of \cite{shepherd_temporally_2009}} 

Again let $P$ be an $X$-program constructed by embedding a submatrix with the properties described above.
Let $\boldsymbol{d},\boldsymbol{e}$ be two bitstrings of length $n$ (the length of a row of $P$). 
Define $P_{\boldsymbol{d},\boldsymbol{e}}$ as the matrix generated by deleting the rows of $P$ orthogonal to $\boldsymbol{d}$ or $\boldsymbol{e}$.
\footnote{In \cite{shepherd_temporally_2009}, $P_{\boldsymbol{d},\boldsymbol{e}}$ is written as $P_{\boldsymbol{d}}\cap P_{\boldsymbol{e}}$.}
Let $\boldsymbol{y}=\sum_{\boldsymbol{p}_{i}\in P_{\boldsymbol{d},\boldsymbol{e}}}\boldsymbol{p}_{i}$ be the vector sum of the rows of $P_{\boldsymbol{d},\boldsymbol{e}}$.
Letting $\boldsymbol{Y}$ be the random variable representing the distribution of $\boldsymbol{y}$ when $\boldsymbol{d}$ and $\boldsymbol{e}$ are chosen uniformly at random, then
\begin{equation}
\Pr\left[\boldsymbol{Y}\cdot\boldsymbol{s}=0\right]=\nicefrac{3}{4}
\end{equation}
\end{fact}

\begin{proof}
(From~\cite{shepherd_temporally_2009})
With $\boldsymbol{y}$ defined as above, we have
\begin{equation}
\boldsymbol{y}\cdot\boldsymbol{s}=\sum_{\boldsymbol{p}_{i}\in P_{\boldsymbol{d},\boldsymbol{e}}}\boldsymbol{p}_{i}\cdot\boldsymbol{s}
\end{equation}

By defintion, $\boldsymbol{p}_{i}\cdot\boldsymbol{s}=1$ if $\boldsymbol{p}_{i} \in P_{\boldsymbol{s}}$.
Therefore $\boldsymbol{y}\cdot\boldsymbol{s}$ is equivalent to simply counting the number of rows common to $P_{\boldsymbol{s}}$ and $P_{\boldsymbol{d},\boldsymbol{e}}$, or equivalently, counting the rows in $P_{\boldsymbol{s}}$ for which $\boldsymbol{p}\cdot\boldsymbol{d}$ and $\boldsymbol{p}\cdot\boldsymbol{e}$ are both 1.
We can express this using the matrix-vector products of $P_{\boldsymbol{s}}$ with $\boldsymbol{d}$ and $\boldsymbol{e}$:

\begin{align}
\boldsymbol{y}\cdot\boldsymbol{s} & =\sum_{\boldsymbol{p}_{i}\in P_{\boldsymbol{s}}}\left(\boldsymbol{p}\cdot\boldsymbol{d}\right)\left(\boldsymbol{p}\cdot\boldsymbol{e}\right)\\
 & =\left(P_{\boldsymbol{s}}\ \boldsymbol{d}\right)\cdot\left(P_{\boldsymbol{s}}\ \boldsymbol{e}\right)
\end{align}

Considering that $P_{\boldsymbol{s}}$ is the generator matrix for an error correcting code, denote $\boldsymbol{c}_{\boldsymbol{d}}=P_{\boldsymbol{s}}\ \boldsymbol{d}$ as the encoding of $\boldsymbol{d}$ under $P_{\boldsymbol{s}}$.
Then we have
\begin{align}
\boldsymbol{y}\cdot\boldsymbol{s} & =\boldsymbol{c}_{\boldsymbol{d}}\cdot\boldsymbol{c}_{\boldsymbol{e}}
\end{align}

Now, note that if a code has $\mathrm{wt}(\boldsymbol{c}) \in \{-1, 0\} \pmod 4$ for all codewords $\boldsymbol{c}$, the extended version of that code (created by adding a single parity bit) is doubly even, that is, has all codeword weights exactly $0 \pmod 4$.
A doubly even binary code is necessarily self-dual, meaning all its codewords are orthogonal.
This implies that any two codewords $\boldsymbol{c}_{\boldsymbol{d}}$ and $\boldsymbol{c}_{\boldsymbol{e}}$ of the original (non-extended) code have $\boldsymbol{c}_{\boldsymbol{d}}\cdot\boldsymbol{c}_{\boldsymbol{e}}=0$ iff either $\boldsymbol{c}_{\boldsymbol{d}}$ or $\boldsymbol{c}_{\boldsymbol{e}}$ has even parity.
Half of our code's words have even parity and $\boldsymbol{c}_{\boldsymbol{d}}$ and $\boldsymbol{c}_{\boldsymbol{e}}$ are random codewords, so the probability that either of them has even parity is $\nicefrac{3}{4}$.
Thus, the probability that $\boldsymbol{y}\cdot\boldsymbol{s}=0$ is $\nicefrac{3}{4}$, proving the fact.
\end{proof}

In the next section, we show that the classical strategy just described can be improved.

\section{Algorithm\label{sec:Algorithm}}

The classical strategy described in Fact~\ref{fact:classical-strategy} above generates vectors that are orthogonal to $\boldsymbol{s}$ with probability $\nicefrac{3}{4}$.
The key to classically defeating the protocol is that it is possible to \emph{correlate} the vectors generated by that strategy, such that there is a non-negligible probability of generating a large set of vectors that \emph{all} are orthogonal to $\boldsymbol{s}$.
These vectors form a system of linear equations that can be solved to yield $\boldsymbol{s}$. 
Finally, with knowledge of $\boldsymbol{s}$ it is trivial to generate samples that pass the verifier's test. 

We follow a modified version of the classical strategy of Fact~\ref{fact:classical-strategy} to generate each vector in the correlated set. 
Crucially, instead of choosing random bitstrings for both $\boldsymbol{d}$ and $\boldsymbol{e}$ each time, we generate a single random bitstring $\boldsymbol{d}$ and hold it constant, only choosing new random values for $\boldsymbol{e}$ with each iteration.
If the encoding $\boldsymbol{c}_{\boldsymbol{d}}$ of $\boldsymbol{d}$ under $P_{\boldsymbol{s}}$ has even parity, \emph{all} of the generated vectors $\boldsymbol{m}_{i}$ will have $\boldsymbol{m}_{i}\cdot\boldsymbol{s}=0$ (see Theorem~\ref{thm:correctprob} below).
This occurs with probability $\nicefrac{1}{2}$ over our choice of $\boldsymbol{d}$.

In practice, it is more convenient to do the linear solve if all $\boldsymbol{m}_{i}\cdot\boldsymbol{s}=1$ instead of 0.
This can be easily accomplished by adding to each $\boldsymbol{m}_{i}$ a vector $\boldsymbol{m}^{*}$ with $\boldsymbol{m}^{*}\cdot\boldsymbol{s}=1$.
It turns out that $\boldsymbol{m}^{*}=\sum_{\boldsymbol{p}\in\mathrm{rows}\left(P\right)}\boldsymbol{p}$ has this property; see proof of Theorem \ref{thm:correctprob}.

The explicit algorithm for extracting the vector $\boldsymbol{s}$ is given in Algorithm \ref{alg:extract}.

\begin{algorithm}
\begin{enumerate}
\item Let $\boldsymbol{m}^{*}=\sum_{\substack{\boldsymbol{p}\in\mathrm{rows}\left(P\right)}
}\boldsymbol{p}$.
\item Pick $\boldsymbol{d}\pick\left\{ 0,1\right\} ^{n}$.
\item Generate a large number (say $2n$) of vectors $\boldsymbol{m}_{i}$ via the following steps, collecting the results into the rows of a matrix $M$.

\begin{enumerate}
\item Pick $\boldsymbol{e}\pick\left\{ 0,1\right\} ^{n}$
\item Let $\boldsymbol{m}_{i}=\boldsymbol{m}^{*}+\sum_{\substack{\boldsymbol{p}\in\mathrm{rows}\left(P\right)\\
\boldsymbol{p}\cdot\boldsymbol{d}=\boldsymbol{p}\cdot\boldsymbol{e}=1
}
}\boldsymbol{p}$
\end{enumerate}

\item Via linear solve, find the set of vectors $\left\{ \boldsymbol{s}_{i}\right\} $ satisfying $M\boldsymbol{s}_{i}=\boldsymbol{1}$, where $\boldsymbol{1}$ is the vector of all ones.
\item For each candidate vector $\boldsymbol{s}_{i}$:
\begin{enumerate}
\item Extract $P_{\boldsymbol{s}_{i}}$ from $P$ by deleting the rows of $P$ orthogonal to $\boldsymbol{s}_{i}$
\item If adding a parity bit to each of the columns $\boldsymbol{c}$ of $P_{\boldsymbol{s}_{i}}$ yields the generator matrix for a code that is doubly even (all basis codewords are doubly even and mutually orthogonal), return $\boldsymbol{s}$ and exit.
\end{enumerate}

\item No candidate vector $\boldsymbol{s}$ was found; return to step 2.
\end{enumerate}

\caption{\label{alg:extract}\textsc{ExtractKey}$\left(P\right)$\protect \\
The algorithm to extract the secret vector $\boldsymbol{s}$ from an $X$-program $P$. $n$~is the number of columns in the $X$-program, and $\protect\pick$ means ``select uniformly from the set.''}

\end{algorithm}

\subsection{Analysis\label{subsec:Analysis}}

In this section we present a theorem and an empirical claim which demonstrate together that Algorithm~\ref{alg:extract} can be used to efficiently extract the key from any $X$-program constructed according to the protocol described in Section~\ref{sec:Background}.
The theorem shows that with probability~$1/2$ a single iteration of the algorithm finds the vector $\boldsymbol{s}$.
The empirical claim is that Algorithm~\ref{alg:extract} is efficient.

\begin{thm}
\label{thm:correctprob}On input an $X$-program $P$ containing a unique submatrix $P_{\boldsymbol{s}}$ with the properties described in Section~\ref{sec:Background}, a single iteration of Algorithm~\ref{alg:extract} will output the vector $\boldsymbol{s}$ corresponding to $P_{\boldsymbol{s}}$ with probability $\frac{1}{2}$.
\end{thm}

\begin{proof}
If $\boldsymbol{s}$ is contained in the set $\left\{ \boldsymbol{s}_{i}\right\} $ generated in step 4 of the algorithm, the correct vector $\boldsymbol{s}$ will be output via the check in step 5 because there is a unique submatrix $P_{\boldsymbol{s}}$ with codewords having $\mathrm{wt}(\boldsymbol{c}) \in \{-1, 0\} \pmod 4$.
$\boldsymbol{s}$ will be contained in $\left\{ \boldsymbol{s}_{i}\right\} $ as long as $M$ satisfies the equation $M\boldsymbol{s}=\boldsymbol{1}$.
Thus the proof reduces to showing that $M\boldsymbol{s}=\boldsymbol{1}$ with probability $\nicefrac{1}{2}$. 

Each row of $M$ is 
\begin{equation}
\boldsymbol{m}_{i}=\boldsymbol{m}^{*}+\bar{\boldsymbol{m}}_{i}
\end{equation}
for a vector $\bar{\boldsymbol{m}}_{i}$ defined as
\begin{equation}
\bar{\boldsymbol{m}}_{i}=\sum_{\substack{\boldsymbol{p}\in\mathrm{rows}\left(P\right)\\
\boldsymbol{p}\cdot\boldsymbol{d}=\boldsymbol{p}\cdot\boldsymbol{e}=1
}
}\boldsymbol{p}
\end{equation}
Here we will show that $\boldsymbol{m}^{*}\cdot\boldsymbol{s}=1$ always and $\bar{\boldsymbol{m}}_{i}\cdot\boldsymbol{s}=0$ for all $i$ with probability $\nicefrac{1}{2}$, implying that $M\boldsymbol{s}=\boldsymbol{1}$ with probability $\nicefrac{1}{2}$.

First we show that $\boldsymbol{m}^{*}\cdot\boldsymbol{s}=1$. $\boldsymbol{m}^{*}$ is the sum of all rows of $P$, so we have
\begin{equation}
\boldsymbol{m}^{*}\cdot\boldsymbol{s}=\sum_{\substack{\boldsymbol{p}\in\mathrm{rows}\left(P\right)}
}\boldsymbol{p}\cdot\boldsymbol{s}=\sum_{\substack{\boldsymbol{p}\in\mathrm{rows}\left(P_{\boldsymbol{s}}\right)}
}1
\end{equation}
We see that the inner product is equal to the number of rows in the submatrix $P_{\boldsymbol{s}}$ (mod 2).
This submatrix is a generator matrix for a code of block size $7 \pmod 8$; thus the number of rows is odd and
\begin{equation}
\boldsymbol{m}^{*}\cdot\boldsymbol{s}=1
\end{equation}

Now we turn to showing that $\bar{\boldsymbol{m}}_{i}\cdot\boldsymbol{s}=0$ for all $i$ with probability $\nicefrac{1}{2}$. 
In the proof of Fact \ref{fact:classical-strategy}, it was shown that for any two vectors $\boldsymbol{d}$ and $\boldsymbol{e}$, vectors $\bar{\boldsymbol{m}}_{i}$ generated by summing rows $\boldsymbol{p}_{i}$ of $P$ for which $\boldsymbol{d}\cdot\boldsymbol{p}_{i}=\boldsymbol{e}\cdot\boldsymbol{p}_{i}=1$
have
\begin{equation}
\bar{\boldsymbol{m}}_{i}\cdot\boldsymbol{s}=0\text{ iff }\boldsymbol{c}_{\boldsymbol{d}}\text{ or }\boldsymbol{c}_{\boldsymbol{e}}\text{ has even parity}\label{eq:mbarparity}
\end{equation}
where $\boldsymbol{c}_{\boldsymbol{d}}$ and $\boldsymbol{c}_{\boldsymbol{e}}$ are the encodings under $P_{\boldsymbol{s}}$ of $\boldsymbol{d}$ and $\boldsymbol{e}$ respectively.
If $\boldsymbol{d}$ is held constant for all $i$, and $\boldsymbol{d}$ happened to be chosen such that
$\boldsymbol{c}_{\boldsymbol{d}}=P_{\boldsymbol{s}}\ \boldsymbol{d}$ has even parity, then $\bar{\boldsymbol{m}}_{i}\cdot\boldsymbol{s}=0$ for all $i$ by Equation \ref{eq:mbarparity}.
Because half of the codewords have even parity, for $\boldsymbol{d}$ selected uniformly at random we have $\bar{\boldsymbol{m}}_{i}\cdot\boldsymbol{s}=0$ for all $i$ with probability $\nicefrac{1}{2}$.

We have shown that $\boldsymbol{m}^{*}\cdot\boldsymbol{s}=1$ always and $\bar{\boldsymbol{m}}_{i}\cdot\boldsymbol{s}=0$ for all $i$ with probability $\nicefrac{1}{2}$.
Therefore we have
\[
\Pr_{\boldsymbol{d}}\left[\boldsymbol{m}_{i}\cdot\boldsymbol{s}=1\ \forall\ i\right]=\nicefrac{1}{2}
\]
Thus $M\boldsymbol{s}=\boldsymbol{1}$ with probability $\nicefrac{1}{2}$.
The algorithm will output $\boldsymbol{s}$ whenever $M\boldsymbol{s}=\boldsymbol{1}$, proving the theorem.
\end{proof}
\begin{center}
\rule[0.5ex]{0.6\columnwidth}{1pt}
\par\end{center}

Before we move on, we remark that while Theorem~\ref{thm:correctprob} treats $X$-programs containing a single unique submatrix with the relevant properties, the algorithm can easily be modified to return the vectors $\boldsymbol{s}$ corresponding to \emph{all} such submatrices, if more exist, by simply accumulating all vectors $\boldsymbol{s}$ for which the check in Step 5(b) succeeds.
We do note, however, that for the protocol described in Section~\ref{sec:Background}, the probability of ``extraneous'' submatrices other than the one intentionally built into the matrix arising by chance is vanishingly small---corresponding to the probability that a random binary linear code happens to be doubly even and self-dual, which is bounded from above by $1/4^{n}$.

Now, having established that each iteration of the algorithm outputs $\boldsymbol{s}$ with probability 1/2, we now turn to analyzing its runtime.
\begin{claim}
\label{claim:runtime}(empirical) Algorithm \ref{alg:extract} halts in $\mathcal{O}\left(n^{3}\right)$ time on average.
\end{claim}
All steps of the algorithm except for step 5 have $\mathcal{O}\left(n^{3}\right)$ scaling by inspection.
The obstacle preventing Claim \ref{claim:runtime} from trivially holding is that it is hard to make a rigorous statement about how large the set of candidate vectors $\left\{ \boldsymbol{s}_{i}\right\} $ is.
Because $\left|\left\{ \boldsymbol{s}_{i}\right\} \right|=2^{n-\mathrm{rank}\left(M\right)}$, we'd like to show that on average, the rank of $M$ is close to or equal to $n$.
It seems reasonable that this would be the case: we are generating the rows of $M$ by summing rows from $P$, and $P$ must have full rank because it contains a rank-$n$ error correcting code.
But the rows of $P$ summed into each $\boldsymbol{m}_{i}$ are not selected independently---they are always related via their connection to the vectors $\boldsymbol{d}$ and $\boldsymbol{e}$, and it's not clear how these correlations affect the linear independence of the resulting $\boldsymbol{m}_{i}$.

\begin{figure}
\begin{centering}
\includegraphics[width=0.8\columnwidth]{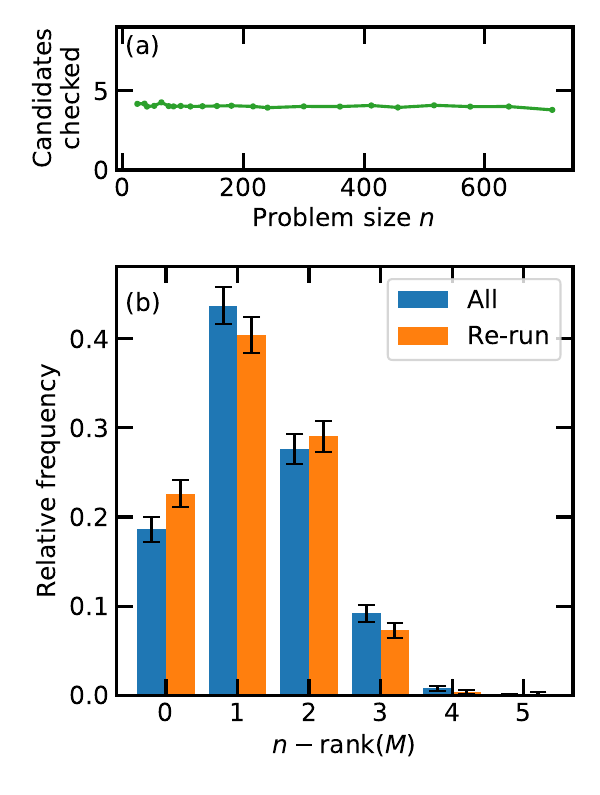}
\par\end{centering}
\caption{
\label{fig:rankM}
\textbf{(a)} The average number of candidate vectors checked before the secret vector $\boldsymbol{s}$ was found, when the algorithm was applied to 1000 unique $X$-programs at each problem size tested.
We observe that the number of vectors to check is qualitatively constant in $n$.
\textbf{(b) }The number of unconstrained degrees of freedom $n-\mathrm{rank}\left(M\right)$ for matrices $M$ generated in step 3 of Algorithm \ref{alg:extract}, for ``good'' choices of $\boldsymbol{d}$ such that $M\boldsymbol{s}=\boldsymbol{1}$.
The rapidly decaying tail qualitatively implies that it is rare for any more than a few degrees of freedom to remain unconstrained.
The blue bars represent the distribution over 1000 unique $X$-programs of size $n=245$.
The algorithm was then re-run on the $X$-programs that had $n-\mathrm{rank}\left(M\right)>4$ to generate the orange bars.}

\end{figure}

Despite the lack of a proof, empirical evidence supports Claim~\ref{claim:runtime} when the algorithm is applied to $X$-programs generated in the manner described in Section~\ref{sec:Background}.
Figure \ref{fig:rankM}(a) shows the average number of candidate keys checked by the algorithm before $\boldsymbol{s}$ is found, as a function of problem size.
The value is constant, demonstrating that the average size of the set $\left\{ \boldsymbol{s}_{i}\right\} $ does not scale with $n$.
Furthermore, the value is small---only about 4.
This implies that $M$ usually has high rank.
In Figure \ref{fig:rankM}(b) we plot explicitly the distribution of the rank of the matrix $M$ over 1000 runs of the algorithm on unique $X$-programs of size $n=245$.
The blue bars (on the left of each pair) show the distribution over all $X$-programs tested, and the sharply decaying tail supports the claim that low-rank $M$ almost never occur. 

A natural next question is whether there is some feature of the $X$-programs in that tail that causes $M$ to be low rank. 
To investigate that question, the algorithm was re-run 100 times on each of the $X$-programs that had $n-\mathrm{rank}\left(M\right)>4$ in the blue distribution.
The orange bars of Figure \ref{fig:rankM}(b) (on the right of each pair) plot the distribution of $n-\mathrm{rank}\left(M\right)$ for that second run.
The similarity of the blue and orange distributions suggests that the rank of $M$ is not correlated between runs; that is, the low rank of $M$ in the first run was not due to any feature of the input $X$-programs.
From a practical perspective, this data suggests that if the rank of $M$ is found to be unacceptably low, the algorithm can simply be re-run with new randomness and the rank of $M$ is likely to be higher the second time.

\subsection{Implementation\label{subsec:Implementation}}

An implementation of Algorithm \ref{alg:extract} in the programming language Julia (along with the code to generate the figures in this manuscript) is available online~\cite{zenodo_code}.
Figure \ref{fig:solvetime} shows the runtime of this implementation for various problem sizes.
Experiments were completed using one thread on an Intel 8268 ``Cascade Lake'' processor.

Note that Figure \ref{fig:solvetime} shows $\mathcal{O}\left(n^{2}\right)$ scaling, rather than $\mathcal{O}\left(n^{3}\right)$ from Claim \ref{claim:runtime}.
This is due to data-level parallelism in the implementation. $\mathbb{Z}_{2}^{n}$ vectors are stored as the bits of 64-bit integers, so operations like vector addition can be performed on 64 elements at once via bitwise operations.
Furthermore, with AVX SIMD CPU instructions, those operations can be applied to multiple 64-bit integers in one CPU cycle.
Thus, for $n$ of order 100, the ostensibly $\mathcal{O}\left(n\right)$ vector inner products and vector sums are performed in constant time, removing one factor of $n$ from the runtime.
The tests in Figure \ref{fig:solvetime} were performed on a CPU with 512 bit vector units.

\section{Discussion\label{sec:Discussion}}

\paragraph{Modifications to the protocol}

A natural question is whether it is possible to modify the original protocol such that this attack is not successful. 
Perhaps $P$ can be engineered such that either 1) it is not possible to generate a large number of vectors that all have a known inner product with $\boldsymbol{s}$, or 2) the rank of the matrix $M$ formed by these generated vectors will never be sufficiently high to allow solution of the linear system.

For 1), our ability to generate many vectors orthogonal to $\boldsymbol{s}$ relies on the fact that the code generated by the hidden submatrix $P_{\boldsymbol{s}}$ has codewords $\boldsymbol{c}$ with $\mathrm{wt}(\boldsymbol{c}) \in \{-1, 0\} \pmod 4$, as shown in the proof of Theorem~\ref{thm:correctprob}.
Unfortunately, this property regarding the weights of the codewords is precisely what gives the quantum sampling algorithm its bias toward generating vectors with $\boldsymbol{x} \cdot \boldsymbol{s} = 0$ (see Fact~\ref{fact:quantum-strategy}).
This fact seems to preclude the possibility of simply removing the special property of the submatrix $P_{\boldsymbol{s}}$ to prevent the attack.

For 2), the main obstacle is that the matrix $P$ \emph{must} have rank $n$ because embedded in it is a code of rank $n$.
The only hope is to somehow engineer the matrix such that linear combinations generated in the specific way described above will not themselves be linearly independent.
It is not at all clear how one would do that, and furthermore, adding structure to the previously-random extra rows of $P$ runs the risk of providing even more information about the secret vector $\boldsymbol{s}$.
Perhaps one could \emph{prove} that the rank of $M$ will be large even for worst-case inputs $P$---this could be an interesting future direction.

\paragraph{Protocols with provable hardness}

The attack described in this paper reiterates the value of building protocols for which passing the test itself, rather than just simulating the quantum device, can be shown to be hard under well-established cryptographic assumptions.
In the past few years, a number of new trapdoor claw-free function based constructions have been proposed for demonstrating quantum computational advantage~\cite{brakerski_cryptographic_2019, brakerski2020simpler, kahanamoku2021classically, alnawakhtha_lattice-based_2022}, as well as some based on other types of cryptography~\cite{yamakawa_verifiable_2022, kalai_quantum_2022}.
Unfortunately, such rigorous results come with a downside, which is an increase in the size and complexity of circuits that must be run on the quantum device.
Exploring simplified protocols that are provably secure is an exciting area for further research.

\paragraph{The \$25 challenge}

When the protocol was first proposed in \cite{shepherd_temporally_2009}, it was accompanied by an internet challenge.
The authors posted a specific instance of the matrix $P$, and offered \$25 to anyone who could send them samples passing the verifier's check.
The secret vector $\boldsymbol{s}$ corresponding to their challenge matrix $P$ is (encoded as a base-64 string):
\begin{verbatim}
BilbHzjYxrOHYH4OlEJFBoXZbps4a54kH8flrRgo/g==
\end{verbatim}
The key was extracted using the implementation of Algorithm~\ref{alg:extract} described in Section~\ref{subsec:Implementation}.

Shepherd and Bremner, the authors of the challenge, have graciously confirmed that this indeed is the correct key.

\paragraph{Summary and outlook}

Here we have described a classical algorithm that passes the interactive quantum test described in \cite{shepherd_temporally_2009}.
We have proven that a single iteration of the algorithm will return the underlying secret vector with probability $\nicefrac{1}{2}$, and empirically shown that it is efficient.
The immediate implication of this result is that the protocol in its original form is no longer effective as a test of quantum computational power.
While it may be possible to reengineer that protocol to thwart this attack, this paper reiterates the value of proving the security of the verification step.
Furthermore, while protocols for quantum advantage with provable classical hardness are valuable in their own right, they can also be used as building blocks for achieving new, more complex cryptographic tasks, like certifiable random number generation, secure remote state preparation, and even the verification of arbitrary quantum computations~\cite{brakerski_cryptographic_2019, alex2019computationallysecure, mahadev_classical_2018}.
As quantum hardware continues to improve and to surpass the abilities of classical machines, quantum cryptographic tools will play an important role in making quantum computation available
as a service.
Establishing the security of these protocols is an important first step.

\paragraph{Acknowledgements}

The author was supported in this work by the National Defense Science and Engineering
Graduate Fellowship (NDSEG).

\bibliographystyle{quantum}
\bibliography{refs}

\end{document}